\documentclass[envcountsame,runningheads]{llncs}
\usepackage{pscproc2}
\usepackage{graphics}

\usepackage[boxruled,vlined]{algorithm2e}
\usepackage{arydshln}

\usepackage{slashbox}
\usepackage{tikz}
\usetikzlibrary{matrix}

\begin{document}
\title{Sorting suffixes of a text \\ via its Lyndon Factorization\thanks{Submitted to the Prague Stringology Conference 2013 (PSC 2013)}}

\author{Sabrina Mantaci, Antonio Restivo,\\
Giovanna Rosone and Marinella Sciortino}
\institute{
University of Palermo, Dipartimento di Matematica e Informatica, Italy\\
\email{\{sabrina,restivo,giovanna,mari\}@math.unipa.it}
}

\newcommand{\ppom} {\preceq_\omega}
\newcommand{\pplex}{\leq_{lex}}

\newcommand{\pom} {\prec_\omega}
\newcommand{\plex}{<_{lex}}

\def\lbwt(#1){{bwt}(#1)}
\def\lbwt{{bwt}}

\def\bigT{W}

\def\bwtS#1{\textsf{bwt}_{#1}(\mathcal{S})}
\def\lcpS#1{\textsf{lcp}_{#1}(\mathcal{S})} 
\def\bwtw#1{\textsf{bwt}_{#1}({w})}
\def\lcpw#1{\textsf{LCP}_{#1}({w})}
\def\rank{\textsf{rank}}
\def\bigS{\mathcal{S}}
\def\BCR{\texttt{BCR}}
\def\BWT{BWT} 
\def\LCP{LCP} 
\def\SA{SA} 
\def\GSA{GSA} 
\def\BCRLCP{\texttt{extLCP}}
\def\EMBWT{\texttt{BCRext} }
\def\EMBWTpp{\texttt{BCRext++} }
\def\BWTE{\texttt{bwte}}
\def\sort#1{\textrm{sort}(#1)}

\newcommand{\fbwt}{{\mathcal B}{\mathcal W}{\mathcal T}}

\maketitle

\begin{abstract}
The process of sorting the suffixes of a text plays a fundamental role in Text Algorithms. They are used for instance in the constructions of the Burrows-Wheeler transform and the suffix array, widely used in several fields of Computer Science. For this reason, several recent researches have been devoted to finding new strategies to obtain effective methods for such a sorting. In this paper we introduce a new methodology in which an important role is played by the Lyndon factorization, so that the local suffixes inside factors detected by this factorization keep their mutual order when extended to the suffixes of the whole word. This property suggests a versatile technique that easily can be adapted to different implementative scenarios.
\end{abstract}


\section{Introduction}

The sorting of the suffixes of a text plays a fundamental role in Text Algorithms with several applications in many areas of Computer Science and Bioinformatics.
For instance, it is a fundamental step, in implicit or explicit way, for the construction of the suffix array ($SA$) and the Burrows-Wheeler Transform ($bwt$).
The $SA$, introduced in 1990 (cf. \cite{Manber:1990}), is a sorted array of all suffixes of a string, where the suffixes are identify by using their positions in the string. Several strategies that privilege the efficiency of the running time or the low memory consumption have been widely investigated (cf. \cite{Puglisi:2007,GrossiSurvey2011}).
The $bwt$, introduced in $1994$ (cf. \cite{bwt94}), permutes the letters of a text according to the sorting of its cyclic rotations, making the text more compressible (cf. \cite{bookBWTAdjeroh:2008}). A recent survey on the combinatorial properties that guarantee such a compressibility after the application of $bwt$ can be found in \cite{RosoneSciortino_CiE2013} (cf. also \cite{RestivoRosoneTCS2011}).
Moreover, in the last years the $SA$ and the $bwt$, besides being important tools in Data Compression, have found many applications well beyond its original purpose (cf. \cite{AbouelhodaKurtzOhlebusch2002,Ferragina:2000,FerraginaManzini2001,MantaciRRS08,Simpson2010,CoxJakobiRosoneST2012,bookBWTAdjeroh:2008}).

The goal of this paper is to introduce a new strategy for the sorting of the suffixes of a word that opens new scenarios of the computation of the $SA$ and the $bwt$.

Our strategy uses a well known factorization of a word $\bigT$ called the \emph{Lyndon factorization} and is based on a combinatorial property proved in this paper, that allows to sort the suffixes of $\bigT$ (``global suffixes'') by using the sorting of the suffixes inside each block of the decomposition (``local suffixes'').

The Lyndon factorization is based on the fact that any word $\bigT$ can be written uniquely as $\bigT=L_1L_2 \cdots L_k$, where
\begin{itemize}
  \item the sequence $L_1,L_2, \ldots, L_k$ is non-increasing with respect to lexicographic order;
  \item each $L_i$ is strictly less than any of its proper cyclic shift (Lyndon words).
\end{itemize}

This factorization was introduced in \cite{ChenFoxLyndon1958} and a linear time algorithm is due to Duval \cite{Duval1983}.
The intuition that the knowledge of Lyndon factorization of a text can be used for the computation of the suffix array of the text itself has been introduced in \cite{BonomoMRRSDLT2013}. Conversely, a way to find the Lyndon factorization from the suffix array can be found in \cite{HohlwegReutenauer2003}.

If $U$ is a factor of a word $\bigT$ we say that the sorting of the local suffixes of $U$ is \emph{compatible} with the sorting of the global suffixes of $\bigT$ if the mutual order of two local suffixes in $U$ is kept when they are extended as global suffixes. The main theorem in this paper states that if $U$ is a concatenation of consecutive Lyndon factors, then the local suffixes in $U$ are compatible with the global suffixes. This suggests some new algorithmic scenarios for the constructions of the $SA$ and the $bwt$. In fact, by performing the Lyndon factorization of a word $\bigT$ by Duval's algorithm, one does not need to get to the end of the whole word in order to start the decomposition into Lyndon factors. Since our result allow to start the sorting of the local suffixes (compatible with the sorting of the global suffixes) as soon as the first Lyndon word is discovered, this may suggest an online algorithm, that do not require to read the entire word to start sorting. Moreover, the independence of the sorting of the local suffixes inside the different Lyndon factors of a text suggests also a possible parallel strategy to sort the global suffixes of the text itself.

In Section \ref{sec:prel} we give the fundamental notions and results concerning combinatorics on words, the Lyndon factorization, the Burrows-Wheeler transform and the suffix array. In Section \ref{sec:method} we first introduce the notion of global suffix on a text and local suffix inside a factor of the text. Then we prove the compatibility between the ordering of local suffixes and the ordering of global suffixes.
In Section \ref{sec:algo} we describe an algorithm that uses the above result to incrementally construct the $bwt$ of a text. Such a method can be also used to explicitly construct the $SA$ of the text.
In Section \ref{sec:conclusion} we discuss about some possible improvements and developments of our method, including implementations in external memory or in place constructions. Finally, we compare our strategy for sorting suffixes with the method proposed in \cite{FerraginaGagieManzini2012} in which a lightweight computation of the $bwt$ of a text is performed by partitioning it into factors having the same length.

\section{Preliminaries}\label{sec:prel}

Let $\Sigma =\{c_1, c_2, \ldots, c_\sigma\}$ be a finite alphabet with $c_1 < c_2 < \ldots < c_\sigma$.
Given a finite word $\bigT=a_1a_2\cdots a_{n}$, $a_i \in \Sigma$ for $i=1,\ldots,n$,
a \emph{factor} of $\bigT$ is written as $\bigT[i,j] = a_i \cdots a_j$. A factor $\bigT[1,j]$ is called a \emph{prefix}, while a factor $\bigT[i,n]$ is called a \emph{suffix}.
In this paper, we also denote by $suf_\bigT(i)$ as the suffix of $\bigT$ starting from position $i$. We omit $\bigT$ when there is no danger of ambiguity.
We say that $x,y\in \Sigma^*$ are {\em conjugate} (or \emph{cyclic shift}) or \emph{$y$ is a conjugate of $x$} if $x=uv$ and $y=vu$ for some $u,v\in \Sigma^*$. Recall that conjugacy is an equivalent relation.

A {\em Lyndon} word is a primitive word which is also the minimum in its conjugacy class, with respect to the lexicographic order relation.
In \cite{Lothaire:2005,Duval1983}, one can find a linear algorithm that for any word $\bigT \in \Sigma^*$ computes the Lyndon word of its conjugacy class.
We call it the Lyndon word of $\bigT$.
Lyndon words are involved in a nice and important factorization property of words.
\begin{theorem}\cite{ChenFoxLyndon1958}
Every word $\bigT \in \Sigma^+$ has a unique factorization $\bigT=L_1L_2 \cdots L_k$ such that $L_1\geq_{lex} \cdots \geq_{lex} L_k$ is a non-increasing sequence of Lyndon words.
\end{theorem}
We call this factorization the \emph{Lyndon factorization}  of a word and it can be computed in linear time (see for instance \cite{Duval1983,Lothaire:2005}).
Duval in \cite{Duval1983} presents two variants of an algorithm of factorization of a word into Lyndon words in time linear in the length of the word.
The first variant of the algorithm uses only three variables for a complete computation and it requires no more than $2n$ comparisons between two letters.
The second one is slightly faster in that sense that it requires no more than $\frac{3n}{2}$ comparisons but it uses an auxiliary storage of size $\frac{n}{2}$.
The basis idea for both these variants is finding each factor of the decomposition of the word $\bigT$ from left to right by eventually reading a long enough prefix of the next Lyndon factor.

Lyndon factorization has been realized also in parallel (cf. \cite{ApostolicoCrochemore1995}) and in external memory (cf. \cite{RohCrochemoreIliopoulosPark:2008}).

One way to define the Burrows-Wheeler Transform ($bwt$) \cite{bwt94} of a string $\bigT$ of length $n$ (although not the most efficient way to compute it) is to construct all $n$ cyclic shifts of $\bigT$ and sort them lexicographically. The output of $bwt$ consists of the pair ($L$, $I$), where $L$ is the sequence of the last character of each rotation in the sorted list and $I$ is an integer denoting the position of the original word in the list.\\
Another more efficient way consists in the concatenating at the the input string $\bigT$ a symbol $\$$ that is smaller than any other letter.
In this case, the $bwt$ is intuitively described as follows:
given a word $\bigT\in \Sigma^*$, $\lbwt(\bigT)$ is a word obtained by sorting the list of the suffixes of $\bigT\$$ and by concatenating the symbols preceding in $\bigT$ each suffix in the sorted list.
In both the cases, it is an invertible transform, i.e., one can recover the original text from its $bwt$.

Note that, in general, the sorting of the conjugates of a word $\bigT$ and the sorting of the suffixes of a word $\bigT\$$ is different, but, as consequence of the properties of Lyndon words, when the word $\bigT$ is the Lyndon word, then the two sorting coincide (cf. \cite[Lemma 12]{GRS2007}).
A study of the combinatorial aspects that connect these two sorting can be found in \cite{BonomoMRRSDLT2013}. In this study an important role is played by the notion of Lyndon word.

Given a text $\bigT$ of length $n$, the suffix array  (SA) for $\bigT$ is an array of integers of range $1$ to $n+1$ specifying the lexicographic ordering of the suffixes of the string $\bigT$. It will be convenient to assume that $\bigT[n+1] = \$$, where $\$$ is smaller than any other letter.
That is, $SA[j] = i$ if and only if $\bigT[i, n+1]$ is the $j$-th suffix of $\bigT$ in ascending lexicographical order.

\begin{figure}[b]
{\scriptsize
$$
\begin{array}{c|c|ccccccccccccc}
SA    & bwt   &       & \multicolumn{12}{c}{Suffixes}                                                            \\
\hline
12    & s     &       & \$    &       &       &       &       &       &       &       &       &       &       &  \\
2     & m     &       & a     & t     & h     & e     & m     & a     & t     & i     & c     & s     & \$    &  \\
7     & m     &       & a     & t     & i     & c     & s     & \$    &       &       &       &       &       &  \\
10    & i     &       & c     & s     & \$    &       &       &       &       &       &       &       &       &  \\
5     & h     &       & e     & m     & a     & t     & i     & c     & s     & \$    &       &       &       &  \\
4     & t     &       & h     & e     & m     & a     & t     & i     & c     & s     & \$    &       &       &  \\
9     & t     &       & i     & c     & s     & \$    &       &       &       &       &       &       &       &  \\
1     & \$    &       & m     & a     & t     & h     & e     & m     & a     & t     & i     & c     & s     & \$ \\
6     & e     &       & m     & a     & t     & i     & c     & s     & \$    &       &       &       &       &  \\
10    & c     &       & s     & \$    &       &       &       &       &       &       &       &       &       &  \\
2     & a     &       & t     & h     & e     & m     & a     & t     & i     & c     & s     & \$    &       &  \\
8     & a     &       & t     & i     & c     & s     & \$    &       &       &       &       &       &       &  \\
\end{array}
$$
}
\caption{The table of the lexicographically sorted suffixes of the word $mathematics\$$ together the $SA(mathematics\$)$ and the $\lbwt(mathematics\$)$.}\label{fig:bwt}
\end{figure}

For instance, if $\bigT=mathematics$ then $\lbwt(\bigT\$)=smmihtt\$ecaa$ and $SA(\bigT\$)=[12, \ 2, \ 7, \ 10, \ 5, \ 4, \ 9, \ 1, \ 6, \ 10, \ 2, \ 8]$. The table obtained by lexicographically sorting all the suffixes of $\bigT\$$ is depicted in Figure \ref{fig:bwt}.

\section{Local and global suffixes of a text}\label{sec:method}

Let $\bigT \in \Sigma^*$ and let $\bigT =L_1L_2 \cdots L_k$ be its Lyndon Factorization.
For each factor $L_r$, we denote by $first(L_r)$ and $last(L_r)$ the position of the first and the last character, respectively, of the factor $L_r$ in $\bigT$. Let $u$ be a factor of $\bigT$. We denote by $suf_u(i) = \bigT[i, last(u)]$ and we call it \emph{local suffix} at the position $i$ with respect to $u$. Note that $suf_\bigT(i) = \bigT[i, n]$ and we call it \emph{global suffix} of $\bigT$ at the position $i$. We write $suf(i)$ instead of $suf_\bigT(i)$ when there is no danger of ambiguity.

\begin{definition}
Let $\bigT$ be a word and let $u$ be a factor of $\bigT$. We say that the sorting of suffixes of $u$ is \emph{compatible} with the sorting of suffixes of $\bigT$ if for all $i,j$ with $first(u) \leq i < j \leq last(u)$,
$$suf_u(i) < suf_u(j) \iff suf(i) < suf(j).$$
\end{definition}

Notice that in general taken an arbitrary factor of a word $\bigT$, the sorting of its suffixes is not compatible with the sorting of the suffixes of $\bigT$. Consider for instance the word $\bigT=abababb$ and its factor $u=ababa$. Then $suf_u(1)=ababa>a=suf_u(5)$ whereas $suf(1)=abababb<abb=suf(5)$.

\begin{theorem}\label{th:SufOrder}
Let $\bigT\in \Sigma^*$ and let $\bigT=L_1L_2 \cdots L_k$ be its Lyndon factorization.
Let $u=L_r L_{r+1}\cdots L_s$.
Then the sorting of the suffixes of $u$ is compatible with the sorting of the suffixes of $\bigT$.
\end{theorem}
\begin{proof}
Let $i$ and $j$ be two indexes with $i<j$ both contained in $u$. We just need to prove that $suf(i)>suf(j) \iff suf_u(i) > suf_u(j)$.
Let $x=\bigT[j, last(L_s)]$ and $y=\bigT[i, i+|x|-1]$.

Suppose that $suf(i)>suf(j)$. Then $y\geq x$ by the definition of lexicographic order. If $y>x$ there is nothing to prove.
If $x=y$, then $suf_u(j)$ is prefix of $suf_u(i)$, so by the definition of lexicographic order $suf_u(i) > suf_u(j)$.

Suppose now that $suf_u(i) > suf_u(j)$. This means that $y\geq x$. If $y>x$ there is nothing to prove. If $x=y$, the index $i+|x|-1$ is in some Lyndon factor $L_m$ with $r \leq m \leq s$, then $L_r\geq L_m\geq L_s$. We denote  $z = \bigT[i+|x|, last(L_m)]$. Then  $suf(i)=xzL_{m+1}\cdots L_k>xL_{s+1}\cdots L_k=suf(j)$, since $z>L_m$ (because $L_m$ is a Lyndon word) and $L_m\geq L_{s+1}$ (since the factorization is a sequence of non increasing factors). \qed
\end{proof}

The above theorem states, in other words, that mutual order of the suffixes of $\bigT$ starting in two positions $i$ and $j$ is the same as the mutual order of the ``local'' suffixes starting in $i$ and $j$ inside the block obtained as concatenation of the consecutive Lyndon factors including $i$ and $j$.

As particular case, the theorem is also true when the two suffixes start in the same Lyndon factor.

We recall that, if $l_1$ and $l_2$ denote two sorted lists of elements taken from any well ordered set, the operation $merge(l_1, l_2)$ consists in obtaining the sorted list of elements in $l_1$ and $l_2$

A consequence of previous theorem is stated in the following proposition.
\begin{proposition}
Let $sort(L_1L_2\cdots L_l)$ and $sort(L_{l+1}L_{l+2}\cdots L_k)$ denote the sorted lists of the suffixes of  $L_1L_2\cdots L_l$ and the suffixes $L_{l+1}L_{l+2}\cdots L_k$, respectively. Then $sort(L_1L_2  \cdots L_k)=merge(sort(L_1L_2\cdots L_l), sort(L_{l+1}L_{l+2}\cdots L_k))$.
\end{proposition}

This proposition suggests a possible strategy for sorting the list of the suffixes of some word $\bigT$:
\begin{itemize}
\item find the Lyndon decomposition of $\bigT$, $L_1L_2 \cdots L_k$;
\item find the sorted list of the suffixes of $L_1$ and, separately, the sorted list of the suffixes of $L_2$;
\item merge the sorted lists in order to obtain the sorted lists of the suffixes of $L_1L_2$;
\item find the sorted list of the suffixes of $L_3$ and  merge it to the previous sorted list;
\item keep on this way until all the Lyndon factors are processed;
\end{itemize}

This kind of strategy could have several advantages: first of all, one can work online, i.e. one can start sorting suffixes as soon as the first Lyndon factor is individuated. This also allow to integrate the sorting process with the Duval's Algorithm for Lyndon decomposition that outputs Lyndon factors online as well.

The second advantage is that this kind of strategy allows parallelization, since every Lyndon factor can be processed separately for sorting its suffixes. These kind of application would require an efficient algorithm to perform the merging of two sorted lists.

A detailed algorithmic description of this method in order to obtain the $bwt$ of a text is given in next section.

\section{An incremental algorithm to sort suffixes of a text}\label{sec:algo}

In this section we propose an algorithm that incrementally constructs the suffix array $SA$ and the Burrows-Wheeler transform $bwt$ of the text $\bigT$ by using its Lyndon factorization. In particular, here we detail the construction of the $bwt$ but an analogous reasoning can be done in order to obtain the suffix array.
We assume that $L_1L_2 \cdots L_k$ is the Lyndon factorization of the word $\bigT[1,n]$.
So  $L_1 \geq L_2 \geq \ldots \geq L_k$.
Such an hypothesis, although strong, is not restrictive because one can obtain the Lyndon factorization of any word in linear time (cf. \cite{Duval1983,Lothaire:2005}).
As shown in previous section, the hypothesis that $\bigT$ is factorized in Lyndon words suggests to connect the problem to the sorting of the local suffixes of $\bigT$ to the lexicographic sorting of the global suffixes of $\bigT$.

Our algorithm, called {\sc Bwt\_Lynd}, considers the input text $\bigT[1,n]$ as logically partitioned into $k$ blocks, where each block corresponds to a Lyndon word, and computes incrementally the $\lbwt(\bigT\$)$ via $k$ iterations, one per block of $\bigT$.
Each block is examined from right to left so that at iteration $i$ we compute $\lbwt(L_1 \cdots L_{i}\$)$ given $\lbwt(L_1 \cdots L_{i-1}\$)$, $\lbwt(L_i\$)$ and $SA(L_i\$)$.
Remark that the positions in $SA(L_i\$)$ range in $[first(L_i), Last(L_i)+1]$. This means that we sum the amount $|L_1 \cdots L_{i-1}|$ to the values of the usual suffix array of $L_i\$$.

The key point of the algorithm comes from Theorem \ref{th:SufOrder}, because the construction of $\lbwt(L_1 \cdots L_{i}\$)$ from   $\lbwt(L_1 \cdots L_{i-1}\$)$  requires only the insertion of the characters of $L_{i}$ in $\lbwt(L_1 \cdots L_{i-1}\$)$ in the same mutual order as they appear in $\lbwt(L_i\$)$. Note that the character $\$$ that follows $L_i$ is not considered in this operation.

Moreover, such an operation does not modify the mutual order of the characters already lying in $\lbwt(L_1 \cdots L_{i-1}\$)$.

For each block $L_i$ with $i$ ranging from $1$ to $k$, the algorithm {\sc Bwt\_Lynd} executes the following steps:
\begin{enumerate}
  \item\label{it:bwtSA} Compute the $\lbwt(L_{i}\$)$ and $SA(L_{i}\$)$.
  \item\label{it:a} Compute the counter array $G[1, |L_{i}|+1]$ which stores in $G[j]$ the number of suffixes of the string $L_1 \cdots L_{i-1}\$$ which are lexicographically smaller than the $j$-th suffix of $L_{i}\$$.
  \item\label{it:merge} Merge $\lbwt(L_1 \cdots L_{i-1}\$)$ and $\lbwt(L_{i}\$)$ in order to obtain $\lbwt(L_1 \cdots L_{i-1}L_{i}\$)$.
\end{enumerate}

\begin{example}\label{ex:merge}
Let $\bigT = aabcabbaabaabdabbaaabbdc$.
The Lyndon factorization of $\bigT$ is $L_1L_2L_3$, where $L_1=aabcabb>L_2=aabaabdabb>L_3=aaabbdc$.
Figure \ref{fig:1} illustrates how Step \ref{it:merge} of the algorithm works.
Note that the positions of the suffixes in $L_2\$$ (i.e. in $SA(L_2\$)$)  are shifted of $|L_1|=7$ positions.
Notice that in the algorithm {\sc Bwt\_Lynd} we do not actually compute the sorted list of suffixes, but we show it in Figure \ref{fig:1} to ease the comprehension of the algorithm. Moreover, the algorithm can be simply adapt to compute the suffix array of $\bigT$, so in Figure \ref{fig:1} the suffix arrays are also shown.

\begin{figure}[!htb]
{\scriptsize
  $$
    \begin{array}{c|c|c|l}
    \multicolumn{4}{c}{L_1\$}      \\
           & SA    & bwt   & \mbox{Sorted Suffixes} \\
\hline
           & 8     & \underline{b}     & \$ \\
           & 1     & \$    & aabcabb\$ \\
           & 5     & c     & abb\$ \\
           & 2     & a     & abcabb\$ \\
           & 7     & b     & b\$ \\
           & 6     & a     & bb\$ \\
           & 3     & a     & bcabb\$ \\
           & 4     & b     & cabb\$ \\
           &       &       &  \\
           &       &    &  \\
    \multicolumn{4}{c}{L_2\$}      \\
     G     & SA    & bwt   & \mbox{Sorted Suffixes} \\
\hline
\textbf{0} & \textbf{11+7=18} & \textbf{b} & \textbf{\$} \\
\textbf{0} & \textbf{1+7=8} & \textbf{\$} & \textbf{\underline{aabaabdabb\$}} \\
\textbf{2} & \textbf{4+7=11} & \textbf{b} & \textbf{aabdabb\$} \\
\textbf{2} & \textbf{2+7=9} & \textbf{a} & \textbf{abaabdabb\$} \\
\textbf{2} & \textbf{8+7=15} & \textbf{d} & \textbf{abb\$} \\
\textbf{4} & \textbf{5+7=12} & \textbf{a} & \textbf{abdabb\$} \\
\textbf{4} & \textbf{10+7=17} & \textbf{b} & \textbf{b\$} \\
\textbf{5} & \textbf{3+7=10} & \textbf{a} & \textbf{baabdabb\$} \\
\textbf{5} & \textbf{9+7=16} & \textbf{a} & \textbf{bb\$} \\
\textbf{7} & \textbf{6+7=13} & \textbf{a} & \textbf{bdabb\$} \\
\textbf{8} & \textbf{7+7=14} & \textbf{b} & \textbf{dabb\$} \\
    \end{array}%
\qquad
    \begin{array}{c}
        \\
    \Rightarrow
    \\
    \end{array}%
\qquad
\begin{array}{c|c|l}
    \multicolumn{3}{c}{L_1L_2\$}      \\
SA    & bwt   & \mbox{Sorted Suffixes} \\
\hline
\textbf{18} & \textbf{b} & \textbf{\$} \\
\textit{8} & \textit{\underline{b}} & \textit{\underline{aabaabdabb\$}} \\
1     & \$    & aabcabbaabaabdabb\$ \\
\textbf{11} & \textbf{b} & \textbf{aabdabb\$} \\
\textbf{9} & \textbf{a} & \textbf{abaabdabb\$} \\
\textbf{15} & \textbf{d} & \textbf{abb\$} \\
5     & c     & abbaabaabdabb\$ \\
2     & a     & abcabbaabaabdabb\$ \\
\textbf{12} & \textbf{a} & \textbf{abdabb\$} \\
\textbf{17} & \textbf{b} & \textbf{b\$} \\
7     & b     & baabaabdabb\$ \\
\textbf{10} & \textbf{a} & \textbf{baabdabb\$} \\
\textbf{16} & \textbf{a} & \textbf{bb\$} \\
6     & a     & bbaabaabdabb\$ \\
3     & a     & bcabbaabaabdabb\$ \\
\textbf{13} & \textbf{a} & \textbf{bdabb\$} \\
4     & b     & cabbaabaabdabb\$ \\
\textbf{14} & \textbf{b} & \textbf{dabb\$} \\
\end{array}
    $$
}   
\caption{
Iteration $2$ of the computation of the $bwt$ of the text $\bigT = aabcabb|aabaabdabb|aaabbdc$ on the alphabet $\{a,b,c,d\}$.
The two columns represent the $bwt$s before and after the iteration.
Note that the first row (the underlined letter) in the table relative to $L_1\$$ and the second row (the underlined suffix) in the table relative to $L_2\$$ flow into the second row in the table relative to $L_1L_2\$$.
Indeed, the suffix $aabaabdabb\$$  is preceded by the symbol $b$ in $L_1L_2\$$.
We use distinct style fonts for each Lyndon word.
}\label{fig:1}
\end{figure}
\end{example}

Step \ref{it:bwtSA} can be executed in linear time $O(|L_i|)$, if $\lbwt(L_{i}\$)$ and $SA(L_{i}\$)$ are stored in internal memory (see \cite{Puglisi:2007,GrossiSurvey2011}).

During Step \ref{it:a}, the algorithm uses the functions $C$ and $rank$ described as follows.
For any character $x \in \Sigma$, let $C(u,x)$ denote the number of characters in $u$ that are smaller than $x$, and let $rank(u, x, t)$ denote the number of occurrences of $x$ in $u[1,t]$. Such functions have been introduced in \cite{Ferragina:2000} for the FM-index. For sake of simplicity we can firstly construct the array $A[1, |L_{i}|+1]$ which stores in $A[j]$ the number of suffixes of the string $L_1 \cdots L_{i-1}\$$ which are lexicographically smaller than the suffix of $L_{i}\$$ starting at the position $j$.
Remark that we set $A[1]=0$ because $L_i[1,|L_i|]\$$ has the same rank of $\$$ between the suffixes of $L_1 \cdots L_{i-1}\$$ and it is preceded by the same symbol $L_{i-1}(|L_{i-1}|)$ in $L_1\cdots L_{i-1}L_i\$$. Consequently, in our algorithm considers the suffixes $L_i[1,|L_i|]\$$ and the suffix $\$$ (of the string $L_1 \cdots L_{i-1}\$$) as the same suffix.
It is easy to prove that the value $A[|L_i|+1]$ is $0$.
The array $A$ is computed from the position $|L_i|$ to $2$ by using Proposition \ref{prop:arrayA}.

\begin{proposition}\label{prop:arrayA}
Let $j$ be a integer ranging from $|L_i|$ to $2$ and let $A[j+1]$ be the number of suffixes of $L_1 \cdots L_{i-1}\$$ lexicographically smaller than $L_i[j+1,|L_i|+1]$. Let $c$ be the first symbol of the suffix $L_i[j,|L_i|+1]$. Then, $$A[j]=C(\lbwt(L_1 \cdots L_{i-1}\$),c)+rank(\lbwt(L_1 \cdots L_{i-1}\$),c,A[j+1]).$$
\end{proposition}
\begin{proof}

Since $c$ is the first symbol of the suffix $L_i[j,|L_i|+1]$, then $L_i[j,|L_i|+1]=cL_i[j+1,|L_i|+1]$. All the suffixes of $L_1 \cdots L_{i-1}\$$ starting with a symbol smaller than $c$ are lexicographically smaller than $L_i[j,|L_i|+1]$. The number of such suffixes is given by $C(\lbwt(L_1 \cdots L_{i-1}\$),c)$. Let us count now the number of suffixes that starting with $c$ and are smaller than $L_i[j,|L_i|+1]$. This is equivalent to counting how many $c$'s occur in $\lbwt(L_1 \cdots L_{i-1}\$)[1,A[j+1]]$. Such a value is given by $rank(\lbwt(L_1 \cdots L_{i-1}\$),c,A[j+1])$.\qed
\end{proof}

It is easy to verify that we can obtain the array $G$ by using the array $A$ and the suffix array $SA(L_i\$)$,
i.e. $G[i] = A[SA(L_i\$)[i]]$. Note that the array $G$ contains the partial sums of the values of the $gap$ array used in \cite{CrauserF02,FerraginaGagieManzini2012}. However, we could directly compute the array $G$ by using the notion of inverse suffix array $ISA$\footnote{The inverse suffix array $ISA$ of a word $\bigT\$$ is the inverse permutation of $SA$, i.e., $ISA[SA[i]] = i$ for all $i \in [1, |w|+1]$. The value $ISA[j]$ is the lexicographical rank of the suffix starting at the position $j$.}. Step \ref{it:a} could be realized in $O(\sum_{j=1,\ldots,i}{|L_j|})$ time because we can build a data structure supporting $O(1)$ time rank queries over $\lbwt(L_1 \cdots L_{i-1}\$)$. The same time complexity is obtained if the rank queries are executed over $\lbwt(L_i\$)$.

Step \ref{it:merge} uses $G$ to create the new array $\lbwt(L_1 \cdots L_{i}\$)$ by merging $\lbwt(L_{i}\$)$ with the $\lbwt(L_1 \cdots L_{i-1}\$)$ computed at the previous iteration. Such a step implicitly constructs the lexicographically sorted list of suffixes starting in $L_1 \cdots L_{i-1}$ and extending up to end of $L_i$ together with the suffixes of $L_i$. In order to do this we keep the mutual order between the suffixes of $L_1 \cdots L_{i-1}\$$ and $L_i\$$ thanks to Theorem \ref{th:SufOrder}. From the definition of the array $G$, it follows that the first two positions of the array $\lbwt(L_1 \cdots L_{i}\$)$ are the first symbol of $\lbwt(L_{i}\$)$ and the first symbol of $\lbwt(L_1 \cdots L_{i-1}\$)$, respectively.
For $j = 3, \ldots, |L_{i}|$ we copy $G[j]$ values from $\lbwt(L_1 \cdots L_{i-1}\$)$ followed by the value $\lbwt(L_{i}\$)[j]$. It is easy to see that the time complexity of Step \ref{it:merge} is $O(\sum_{j=1,\ldots,i}{|L_j|})$, too.

From the description of the algorithm and by proceeding by induction, one can prove the following proposition.
\begin{proposition}
At the end of the iteration $k$, Algorithm {\sc Bwt\_Lynd} correctly computes $\lbwt(L_1 \cdots L_{k}\$)$. Each iteration $i$ runs in $O(\sum_{j=1,\ldots,i}{|L_j|})$ time. The overall time complexity is $O(k^2M)$, where $M=\max_{i=1,\ldots,k}(|L_i|)$.
\end{proposition}

\section{Discussions and conclusions}\label{sec:conclusion}

The goal of this paper is to propose a new strategy to compute the $bwt$ and the $SA$ of a text by decomposing it into Lyndon factors and by using the compatibility relation between the sorting of its local and global suffixes. At the moment, the quadratic cost of the algorithm could make it impractical. However, from one hand, in order to improve our algorithm, efficient dynamic data structure for the rank operations and for the insertion operations could be used. Navarro and Nekrich's recent result \cite{NNsoda13} on optimal representations of dynamic sequences shows that one can insert symbols at arbitrary positions and compute the rank function in the optimal time $O(\frac{\log n}{\log \log n})$ within essentially $nH_0(s) + O(n)$ bits of space, for a sequence $s$ of length $n$.
On the other hand, our technique, differently from other approaches in which partitions of the text are performed, is quite versatile so that it easily can be adapted to different implementative scenarios.

For instance, in \cite{FerraginaGagieManzini2012} the authors describe an algorithm, called {\sc bwte}, that logically partitions the input text $\bigT$ of length $n$ into blocks of the same length $m$, i.e. $\bigT = T_{n/m}T_{n/m-1} \cdots  T_1$ and computes incrementally the $bwt$ of $\bigT$ via $n/m$ iterations, one per block of $\bigT$.
Text blocks are examined from right to left so that at iteration $h + 1$, they compute and store on disk $bwt(T_{h+1} \cdots T_1)$ given $bwt(T_{h} \cdots T_1)$. In this case the mutual order of the suffixes in each block depends on the order of the suffixes of the next block. Our algorithm {\sc Bwt\_Lynd} builds the $\lbwt$ of a text or its $SA$ by scanning the text \emph{from left to right} and it could run online, i.e. while the Lyndon factorization is realized. One of the advantages is that adding new text to the end does not imply to compute again the mutual order of the suffixes of the text analyzed before, unless for the suffixes of the last Lyndon word that could change by adding characters on the right. Moreover, as described in the previous section, the text could be partitioned into several sequences of consecutive blocks of Lyndon words, and the algorithm can be applied \emph{in parallel} to each of those sequences. Furthermore, also the Lyndon factorization can be performed in parallel, as shown in \cite{ApostolicoCrochemore1995}. Alternatively, since we read each symbol only once, also an in-place computation could be suggested by the strategy proposed in \cite{CGKL_CPM2013}, in which the space occupied by text $\bigT$ is used to store the $bwt(\bigT)$.

Finally, in the description of the algorithm we did not mention the used workspace. In fact, it could depend on the time-space trade-off that one should reach. For instance, the methodologies used in \cite{BauerCoxRosoneTCS2013,FerraginaGagieManzini2012} where disk data access are executed only via sequential scans could be adapted in order to obtain a lightweight version of the algorithm. An external memory algorithm for the Lyndon factorization can be found in \cite{RohCrochemoreIliopoulosPark:2008}.
We remark that that the method proposed in \cite{FerraginaGagieManzini2012} could be integrated into {\sc Bwt\_Lynd} in the sense that one can apply {\sc bwte} to compute at each iteration the $bwt$ and the $SA$ of each block of the Lyndon partition.

In conclusion, our method seems lay out the path towards a new approach to the problem of sorting the suffixes of a text in which partitioning the text by using its combinatorial properties allows it to tackle the problem in local portions of the text in order to extend efficiently solutions to a global dimension.

\bibliographystyle{psc}          


\end{document}